\documentclass{amsart}
\usepackage{amsmath,amssymb,latexsym}

\newcommand{\IN}{\mathbb N}
\newcommand{\C[2]}{{\textstyle{\binom{#1}{#2}}}}
\newcommand{\one}{1}
\newcommand{\zero}{0}

\newtheorem{theorem}{Theorem}
\newtheorem{corollary}{Corollary}

\newtheorem{lemma}{Lemma}
\theoremstyle{remark}
\newtheorem{remark}{Remark}

\title[Toehold Purchase Problem]{Toehold Purchase Problem:\\ A comparative analysis of two strategies}
\author{I.~Banakh, T.~Banakh, M.~Vovk, \boxed{\rm P.~Trisch}}


\begin{document}

\begin{abstract}
Toehold purchase, defined here as purchase of one share in a firm by an investor preparing a tender offer to acquire majority of shares in it, reduces by one the number of shares this investor needs for majority. In the paper we construct mathematical models for the toehold and no-toehold strategies and compare the expected profits of the investor and the probabilities of takeover the firm in both strategies. It turns out that the expected profits of the investor in both strategies coincide. On the other hand, the probability of takeover the firm using the toehold strategy is considerably higher comparing to the no-toehold strategy. In the analysis of the models we apply the apparatus of incomplete Beta functions and some refined bounds for central binomial coefficients.
\end{abstract}
\maketitle

\section{Introduction}
\parskip1pt

This paper is about the toehold purchase problem. By a toehold we mean either the number or the fraction of shares owned by an outside investor considering or preparing a tender offer to acquire majority of shares and take over. By a tender offer we mean a proposal made by an investor to shareholders to tender their shares, with the hope to obtain majority of shares and take over. At the time of such an offer, an investor may already own, say, one-share-toehold. In our model the firm is going to be widely held and each shareholder will own one share. Outside investor will make a tender offer to all shareholders if s/he does not own a toehold and to all shareholders excluding self when s/he does own a toehold. For our purposes we consider the terms `tender offer' and `bid' as synonyms. Sometimes there is an upper bound on the number of possible stake (shareholding) that the outside investor may hold at the time s/he places a tender offer.
Here surfaces one of the questions of toehold literature. If an investor is allowed to hold only a certain fraction of shares when s/he wishes to place a tender offer (but not more), would s/he always want to hold this maximum possible stake? If not, why not? Probably with this question in mind, a number of toehold theories look at optimal toeholds in a variety of settings and under variety of assumptions about market structure, ownership structure (how many shares each shareholder owns), information structure or the number of investors (one, two or more); see Grossman, Hart \cite{GH80}, Bagnoli, Lipman \cite{BL88}, \cite{BL96}, Singh \cite{S98}, Ravid, Spiegel \cite{RS99}, Betton, Eckbo \cite{BE00}, Bris \cite{B02}, Goldman, Qian \cite{GQ05}, Ettinger \cite{E09}, Betton, Eckbo, Thorburn \cite{BET09}, Chatterjee, John, Yan \cite{CJY12}. 

Our approach is different in that we specifically assume that there is only one investor who is considering a tender offer and that if this investor does decide to purchase a toehold then s/he purchases only one share. If there is no toehold, then our assumptions follow the lines of Bagnoli and Lipman \cite{BL88}. If investor purchases a toehold, then the circumstances of the tender offer are different. The difference does not only lie in the fact that the offer is made to one fewer shareholders. In this case investor's tender offer might (and generally would) take into account the effect of potential takeover on the worth of a toehold. Our setting is rudimentary in that there are no asymmetries of information, toehold is one share and key to toehold purchase is either yes or no answer. The two strategies (no-toehold and toehold) of the outside investor are described in Section 2.  The main results of Section 2 are Theorems~\ref{t0} and \ref{t1}. In Theorem 1 we calculate the principal parameters of the non-toehold strategy: the price of a share $X_0$ suggested by the investor in the tender offer, the probability $\sigma_0$ that a shareholder will sell her/his share to the investor, the probability $P_0$ of takeover the firm, and the expected profit $\Pi_0$ of the investor. In Theorem 2 we calculate the respective parameters $X_1$, $\sigma_1$, $P_1$, $\Pi_1$ for the toehold strategy. Comparing the obtained formulas for these parameters we discovered that both strategies yield the same expected profit $\Pi_1=\Pi_0$ and the same probability $\sigma_0=\sigma_1$ that a shareholder will sell her/his share to the investor. On the other hand, the probability $P_1$ of takeover the firm using the toehold strategy is higher than the corresponding probability $P_0$ for the no-toehold strategy. This follows from Theorem~\ref{bounds} that yields some lower and upper bounds on the parameters $X_i$, $P_i$, $\Pi_i$, $i\in\{0,1\}$, of our models. The proof of Theorem~\ref{bounds} (presented in Appendix) is not trivial and uses the mathematical apparatus of incomplete beta functions and some non-trivial bounds on the central binomial coefficients. In Section 4 we make some mathematical conclusions that follow from the analysis of our models.

\section{Models}\label{s:models}

We assume that a firm has $2n+1$ shareholders. Each shareholder owns one share. The
worth of each share, if the firm continues to be run by incumbent management, is normalized to $0$.
There is also an outside investor $\mathbf B$ who is considering takeover bid. If investor takes over, the value of each share is increased to $1$.

Now we consider two strategies of the investor $\mathbf B$ who is willing to take over the firm buying a majority of shares.
\smallskip

${\mathbf 0}.$ The first strategy will be referred to as the {\em no-toehold strategy} and its parameters will be labeled by the subscript $\zero$. Following the no-toehold strategy, the investor $\mathbf B$ makes a tender offer to all $2n+1$ shareholders suggesting a price $X$ for each share. Shareholders decide
independently whether to accept or to reject the tender offer. They may use mixed strategies, i.e. accept
the offer with certain probability $\sigma$. Simple majority of $n+1$ shares is necessary for takeover. Tender offer is unconditional in the sense that if less than $n$ shareholder accept the tender offer, then $\mathbf B$ has to purchase shares from those shareholders who accepted the offer, even though in that case  $\mathbf B$ becomes a minority shareholder, the worth of each share value remains at $0$ and such purchase is ex post unprofitable for  $\mathbf B$  as long as $X>0$.

Suppose shareholders use symmetric mixed strategies, in which in response to tender offer $X$ all of
them accept the tender offer with probability $\sigma\in(0,1)$ and reject it with probability $(1-\sigma)$. For the
pair $(X,\sigma)$ to be equilibrial, each shareholder has to be indifferent between
tendering and not tendering her share, or otherwise she would not use mixing strategy. If she tenders,
she ends up with $X$, and if she does not, her unsold stake is worth more than $0$ if among remaining
$2n$ shareholders at least $n+1$ shareholders tender their shares. That happens with probability
$\sum_{k=n+1}^{2n}\C[2n]{k}\sigma^k(1-\sigma)^{2n-k}$. In that case the firm is taken over. A shareholder who did not tender her
share remains a minority shareholder who ``free-rides'' on investor's improvement in firm value from $0$ to $1$. So the pair $(X,\sigma)$ can be a suspect for a symmetric mixed strategy equilibrium only if
\begin{equation}\label{eq:X}
X=\sum_{k=n+1}^{2n}\C[2n]{k}\sigma^k(1-\sigma)^{2n-k}.
\end{equation}
Here by $$\binom{n}{k}=\frac{n!}{k!(n-k)!}$$we denote the binomial coefficients.

The investor's expected profit $\Pi$ is calculated using three variables: the number of tendered shares, probability that exactly that many shares are tendered, and the share value: $$\Pi=(0-X)\sum_{k=1}^{n}k\C[2n{+}1]{k}\sigma^k(1-\sigma)^{2n{+}1{-}k}+(1{-}X)\sum_{k=n+1}^{2n+1}k\C[2n{+}1]{k} \sigma^k(1-\sigma)^{2n{+}1{-}k}.$$

After a suitable rearrangement and substituting for $X$ the sum (\ref{eq:X}) we obtain:
$$
\begin{aligned}
\Pi&=\sum_{k=n+1}^{2n+1}\C[2n+1]{k}k\,\sigma^k(1-\sigma)^{2n+1-k}- X\sum_{k=1}^{2n+1}\C[2n+1]{k}\,k\,\sigma^k(1-\sigma)^{2n+1-k}=\\
&=(2n{+}1)\kern-5pt\sum_{k{=}n{+}1}^{2n{+}1}\C[2n]{k{-}1}\sigma^k(1{-}\sigma)^{2n{+}1{-}k}{-}X(2n{+}1)\kern-3pt\sum_{k=1}^{2n{+}1}\C[2n]{k{-}1}
\sigma^{k}(1{-}\sigma)^{2n{+}1{-}k}=\\
&=(2n{+}1)\sum_{k=n}^{2n}\C[2n]{k} \sigma^{k+1}(1-\sigma)^{2n{-}k}-X (2n{+}1)\,\sigma\sum_{k=0}^{2n}\C[2n]{k}\sigma^{k}(1-\sigma)^{2n{-}k}=\\
&=(2n{+}1)\sum_{k=n}^{2n}\C[2n]{k} \sigma^{k+1}(1-\sigma)^{2n-k}-(2n{+}1)\,\sigma X(\sigma+(1-\sigma))^{2n}=
\end{aligned}
$$
$$\begin{aligned}
&=(2n{+}1)\sum_{k=n}^{2n}\C[2n]{k} \sigma^{k+1}(1-\sigma)^{2n-k}-(2n{+}1)\sigma X=\\
&=(2n{+}1)\sum_{k=n}^{2n}\C[2n]{k}\sigma^{k+1}(1-\sigma)^{2n-k}-(2n{+}1)\sigma \sum_{k=n+1}^{2n}\C[2n]{k}\sigma^k(1-\sigma)^{2n-k}=\\
&=(2n{+}1)\sum_{k=n}^{2n}\C[2n]{k}\sigma^{k+1}(1-\sigma)^{2n-k}-(2n{+}1) \sum_{k=n+1}^{2n}\C[2n]{k}\sigma^{k+1}(1-\sigma)^{2n-k}=\\
&=(2n+1)\C[2n]{n}\sigma^{n+1}(1-\sigma)^{n}.
\end{aligned}
$$
The maximal value
$$\Pi_\zero=(2n+1)\C[2n]{n}\sigma_\zero^{n+1}(1-\sigma_\zero)^{n}=\C[2n]{n}\frac{(n+1)^{n+1}n^n}{(2n+1)^{2n}}$$of the profit of the investor is attained for the probability $$\sigma_\zero=\frac{n+1}{2n+1}$$ that corresponds to the price of a share
$$X_\zero=\sum_{k=n+1}^{2n}\C[2n]{k}\sigma_\zero^{k}
(1-\sigma_\zero)^{2n-k}=\sum_{k=n+1}^{2n}\C[2n]{k}\frac{(n+1)^{k}n^{2n-k}}{(2n+1)^{2n}}.$$
In this situation the probability of takeover the firm by the investor equals
$$
P_\zero=\sum_{k=n+1}^{2n+1}\C[2n+1]{k}\sigma_\zero^k(1-\sigma_\zero)^{2n+1-k}=
\sum_{k=n+1}^{2n+1}\C[2n+1]{k}\frac{(n+1)^k n^{2n+1-k}}{(2n+1)^{2n+1}}.
$$

The no-toehold strategy will be denoted by $\mathcal S_\zero$. We summarize our description of this strategy in the following:

\begin{theorem}\label{t0} If the investor uses the no-toehold strategy $\mathcal S_0$ to take over a firm with $(2n+1)$ shareholders, then he should offer the price
$$X_\zero=\sum_{k=n+1}^{2n}\C[2n]{k}\frac{(n+1)^kn^{2n-k}}{(2n+1)^{2n}}$$
for a share in the tender offer and can expect to take over the firm with probability
$$P_\zero=\sum_{k=n+1}^{2n+1}\C[2n+1]{k}\frac{(n+1)^kn^{2n+1-k}}{(2n+1)^{2n+1}}$$and expect for the profit
$$
\Pi_\zero=\C[2n]{n}\frac{(n+1)^{n+1}n^n}{(2n+1)^{2n}}.
$$
To maximize their expected profit the shareholders should sell their shares to the investor with probability
$$\sigma_\zero=\frac{n+1}{2n+1}.$$
\end{theorem}
\smallskip

$\mathbf 1$. Now we consider a more complex strategy $\mathcal S_1$ called the {\em toehold strategy}. Following this strategy the investor $\mathbf B$ first tries to purchase one-share 'toehold' from a
shareholder $\mathbf A$ who is aware that $\mathbf B$ is about to launch a tender offer to acquire majority of shares suggesting the price $X_0$ for a share. We
assume that $\mathbf A$ is the only shareholder from whom $\mathbf B$ is able to purchase a toehold, and $\mathbf A$ agrees to sell her share to the investor $\mathbf B$ for the price $X_0$.

After buying the toehold from the shareholder $\mathbf A$, the investor announces a post-toehold tender offer to the remaining $2n$ shareholders, offering a price $X_\one$ for a share.
If $\sigma_\one$ is the probability that a shareholder will tender her share for that price, then the equilibrium will occur if
$$X_\one=\sum_{k=n}^{2n-1}\C[2n-1]{k}\sigma_\one^k(1-\sigma_\one)^{2n-1-k},$$
which is equal to the probability that among $2n-1$ shareholders at least $n$ will sell their shares.

The probability of takeover the firm in the post toehold tender is equal to
$$P_{\one}=\sum_{k=n}^{2n}\C[2n]{k}\sigma_1^k(1-\sigma_1)^{2n-k}$$and the expected profit $\Pi_\one$ of the investor for the toehold strategy is equal to
{
$$
\begin{aligned}
&\Pi_1=({-}X_0{+}1\cdot P_\one){+}(0{-}X_{\one})\sum_{k=1}^{n{-}1}k\C[2n]{k}\sigma_\one^k(1{-}\sigma_\one)^{2n{-}k}{+} (1{-}X_{\one})\sum_{k=n}^{2n}k\C[2n]{k} \sigma_\one^k(1{-}\sigma_\one)^{2n{-}k}=\\
&=({-}X_0{+}P_\one){-}X_{\one}2n\sigma_1\sum_{k=1}^{2n}\C[2n{-}1]{k{-}1}\sigma_\one^{k{-}1}(1{-}\sigma_\one)^{2n{-}k}{+} 2n\sigma_1\sum_{k=n}^{2n}\C[2n{-}1]{k{-}1} \sigma_\one^{k{-}1}(1{-}\sigma_\one)^{2n{-}k}=\\
&=({-}X_0{+}P_\one){-}2n\sigma_1X_1\kern-4pt\sum_{k=0}^{2n{-}1}\kern-2pt\C[2n{-}1]{k}\sigma_\one^{k}(1{-}\sigma_\one)^{2n{-}k{-}1}{+} 2n\sigma_1\kern-6pt\sum_{k=n{-}1}^{2n{-}1}\kern-2pt\C[2n{-}1]{k} \sigma_\one^{k}(1{-}\sigma_\one)^{2n{-}k{-}1}=\\
&=(-X_0+P_\one)-2n\sigma_1X_{\one}(\sigma_1+(1-\sigma_1))^{2n-1}+ 2n\sigma_\one\sum_{k=n-1}^{2n-1}\C[2n-1]{k} \sigma_\one^{k}(1-\sigma_\one)^{2n-k-1}=\\
&=({-}X_0{+}P_\one){-}2n\sigma_1\kern-3pt\sum_{k=n}^{2n{-}1}\C[2n{-}1]{k}\sigma_\one^k(1{-}\sigma_\one)^{2n{-}1{-}k} {+} 2n\sigma_\one\kern-4pt\sum_{k{=}n{-}1}^{2n{-}1}\C[2n{-}1]{k} \sigma_\one^{k}(1{-}\sigma_\one)^{2n{-}k{-}1}=\\
&=(-X_0+P_\one)+2n\C[2n-1]{n-1}\sigma_\one^{n}(1-\sigma_\one)^{n}=(-X_0+P_\one)+n\C[2n]{n}\sigma_\one^{n}(1-\sigma_\one)^{n}=\\
&=-X_0+\sum_{k=n}^{2n}\C[2n]{k}\sigma_1^k(1-\sigma_1)^{2n-k}+n\C[2n]{n}\sigma_\one^n(1-\sigma_\one)^n.
\end{aligned}
$$
}
To find the maximal value of the expected profit $\Pi_1$, consider the derivative
{
$$
\begin{aligned}
\frac{d\Pi_1}{d\sigma_1}&=\frac{d}{d\sigma_1}\sigma_1^{2n}+\frac{d}{d\sigma_1}\sum_{k=n}^{2n-1}\C[2n]{k}\sigma_1^k(1-\sigma_1)^{2n-k}+\frac{d}{d\sigma_1}n\C[2n]{n}\sigma_\one^n(1-\sigma_\one)^n=\\
&=2n\sigma_1^{2n-1}+\sum_{k=n}^{2n-1}\C[2n]{k}\big(k\sigma_1^{k-1}(1{-}\sigma_1)^{2n-k}-(2n-k)\sigma_1^k(1{-}\sigma_1)^{2n-k-1}\big)+n^2\C[2n]{n}\sigma_1^{n{-}1}(1{-}\sigma_1)^{n{-}1}(1{-}2\sigma_1)=\\
&=\sum_{k=n}^{2n}2n\C[2n-1]{k-1}\sigma_1^{k-1}(1-\sigma_1)^{2n-k}-\sum_{k=n}^{2n-1}2n\C[2n-1]{k}\sigma_1^k(1-\sigma_1)^{2n-1-k}+n^2\C[2n]{n}\sigma_1^{n-1}(1-\sigma_\one)^{n-1}(1-2\sigma_1)=\\
&=\sum_{k=n-1}^{2n-1}2n\C[2n-1]{k}\sigma_1^{k}(1{-}\sigma_1)^{2n-k-1}-\sum_{k=n}^{2n-1}2n\C[2n-1]{k}\sigma_1^k(1{-}\sigma_1)^{2n-1-k}+n^2\C[2n]{n}\sigma_1^{n-1}(1{-}\sigma_1)^{n-1}(1-2\sigma_1)=\\
&=2n\C[2n-1]{n-1}\sigma_1^{n-1}(1-\sigma_1)^{n}+n^2\C[2n]{n}\sigma_1^{n-1}(1-\sigma_\one)^{n-1}(1-2\sigma_1)=\\
&=n\C[2n]{n}\sigma_1^{n-1}(1-\sigma_1)^{n-1}(1-\sigma_1)+n^2\C[2n]{n}\sigma_1^{n-1}(1-\sigma_\one)^{n-1}(1-2\sigma_1)=\\
&=n\C[2n]{n}\sigma_1^{n-1}(1-\sigma_1)^{n-1}\big(1-\sigma_1+n(1-2\sigma_1)\big)\\
\end{aligned}
$$}
and observe that it is equal to zero at $\sigma_1=\frac{n+1}{2n+1}=\sigma_0$.

So, for $\sigma_1=\sigma_0=\frac{n+1}{2n+1}$ the expected profit $\Pi_1$ attains its maximal value
{
$$
\begin{aligned}
\Pi_1&=-X_0+\sum_{k=n}^{2n}\C[2n]{k}\sigma_1^k(1-\sigma_1)^{2n-k}+n\C[2n]{n}\sigma_\one^n(1-\sigma_\one)^n=\\
&=-\sum_{k=n+1}^{2n}\C[2n]{k}\sigma_0^k(1-\sigma_0)^{2n-k}+\sum_{k=n}^{2n}\C[2n]{k}\sigma_1^k(1-\sigma_1)^{2n-k}+n\C[2n]{n}\sigma_\one^n(1-\sigma_\one)^n=\\
&=\C[2n]{n}\sigma_1^n(1-\sigma_1)^{n}+n\C[2n]{n}\sigma_\one^n(1-\sigma_\one)^n=(n+1)\C[2n]{n}\frac{(n+1)^nn^n}{(2n+1)^2n}=\Pi_0.
\end{aligned}
$$
}

The above discussion can be summed up in:

\begin{theorem}\label{t1} If the investor follows the toehold strategy $\mathcal S_1$, then he buys a toehold from the shareholder $\mathbf A$ offering the price $X_0$ for her share and then in the post-toehold offer he offers the price
$$X_\one=\sum_{k=n}^{2n-1}\C[2n-1]{k}\frac{(n+1)^kn^{2n-1-k}}{(2n+1)^{2n-1}}$$
for a share, in which case the
shareholders will sell their shares with probability
$$\sigma_\one=\frac{n+1}{2n+1}=\sigma_0,$$
the investor can takeover the firm with probability
$$P_\one=\sum_{k=n}^{2n}\C[2n]{k}\frac{(n+1)^kn^{2n-k}}{(2n+1)^{2n}}$$
and can expect for the profit
$$\Pi_\one=\C[2n]{n}\frac{(n+1)^{n+1}n^n}{(2n+1)^{2n}}=\Pi_0.$$
\end{theorem}

As we see from Theorems~\ref{t0}, \ref{t1}, the no-toehold and toehold strategies yield the same profit $\Pi_\zero=\Pi_\one$ and the same probability $\sigma_0=\sigma_1=\frac{n+1}{2n+1}$ of selling  their shares by the shareholders in the tender offers. On the other hand, the prices for a share and the probabilities $P_\zero$ and $P_\one$ of takeover the firm are different for these two strategies. The precise estimate of the differences $P_\one-P_\zero$ and $X_\one-X_\zero$ will be given in Corollary~\ref{diff}. Now, let us consider a simple example.

\subsection{\bf A firm with 3 shareholders}
In case of 3-shareholders (which corresponds to $n=1$) the values of all parameters from Theorems~\ref{t0} and \ref{t1} can be easily calculated:
\begin{itemize}
\item $\sigma_\zero=2/3$ is the probability that shareholders will sell their shares to the investor for the price:
\item $X_{\zero}=4/9$ suggested by the investor in the no-toehold strategy,
\item $P_{\zero}=20/{27}$ is the probability of taking over the firm in no-toehold strategy;
\item $\Pi_{\zero}={8}/{9}$ is the expected profit of the investor in the no-toehold strategy;
\smallskip

\item $\sigma_{\one}=2/3$ is the probability that a shareholder will tender her share to the investor for the price:
\item $X_{\one}=2/3$ suggested by the investor in the post-toehold tender offer,
\item $P_{\one}=8/9$ is the probability of taking over the firm in the toehold strategy,
\item $\Pi_{\one}=8/9$ is the expected profit of the investor in the toehold strategy.
\end{itemize}
Looking at these data, we see that both strategies yield the same profit but the toehold strategy is much better than the no-toehold strategy in the sense of probability of takeover the firm.

It turns out that the same situation happens for all $n\in\IN$, see Corollary~\ref{diff} below. In this corollary we shall prove that the difference $P_{\one}-P_{\zero}$ of probabilities for the toehold and no-toehold strategies is strictly positive and has order $P_{\one}-P_{\zero}\approx \frac{1}{2\sqrt{\pi n}}$.

\section{Explicit analytic expressions for the parameters of the models}

For deriving the lower and upper bounds presented in Theorem~\ref{bounds} we shall transform the binomial sums appearing in the expressions of the parameters of our models and obtain precise analytic formulas for these parameters, after which we shall evaluate them using some bounds on central binomial coefficients and simple bounds giving by Taylor series. Our principal tool in finding explicit analytic expressions for the parameters of the model is use of incomplete beta functions.

By definition, the {\em beta function} is the function
$$B(a,b)=\int_0^1t^{a-1}(1-t)^{b-1}dt$$depending on two real positive parameters $a,b$.
For fixed $a,b$ the function
$$B_x(a,b)=\int_0^xt^{a-1}(1-t)^{b-1}dt$$on the variable $x\in[0,1]$ is called the {\em incomplete beta function}. A remarkable property of the incomplete beta function is that for positive integer numbers $a,b$ its value is proportional to a tail of the binomial series:
\begin{equation}\label{eq:beta}
\sum_{k=a}^{a{+}b{-}1}\C[a{+}b{-}1]{k}x^k(1-x)^{a{+}b{-}1{-}k}=
a\C[a{+}b{-}1]{a}B_x(a,b)=a\C[a{+}b{-}1]{a}\int_0^{x}t^{a-1}(1-t)^{b{-}1}dt
\end{equation}
This equality plays a fundamental role in our subsequent arguments and will be referred to as the {\em beta-equality}. For the proof of the beta-equality and other information on (incomplete) beta functions, we refer the reader to the survey paper of Dutka \cite{D81}.

Beta functions will be used in the proof of the following theorem that gives explicit analytic formulas for the parameters describing the no-toehold and toehold strategies.

\begin{theorem}\label{formulas} The parameters of the models can be calculated by the following formulas:
\begin{enumerate}
\item The price $X_{\zero}=\sum_{k=n+1}^{2n}\C[2n]{k}\frac{(n+1)^kn^{2n-k}}{(2n+1)^{2n}}$ suggested by the investor in the no-toehold strategy can be found by the formula
$$X_{\zero}=\frac12-\frac1{2^{2n+1}}{\C[2n]{n}}\Big(1-\frac1{(2n+1)^2}\Big)^n+
\frac{n}{2^{2n}}{\C[2n]{n}}\int_0^{\frac1{2n+1}}(1-t^2)^{n-1}dt.$$
\item The probability $P_{\zero}=\sum_{k=n+1}^{2n+1}\C[2n+1]{k}\frac{(n+1)^kn^{2n+1-k}}{(2n+1)^{2n+1}}$ of taking over the firm in the no-toehold strategy can be found as
$$P_{\zero}=\frac12+\frac{(2n+1)}{2^{2n+1}}\C[2n]{n}\int_0^{\frac1{2n+1}}(1-t^2)^{n}dt.$$
\item The expected profits $\Pi_{\zero}=\Pi_\one=\C[2n]{n}\frac{n^n(n+1)^{n+1}}{(2n+1)^{2n}}$ of the investor can be found by
$$\Pi_{\zero}=\Pi_\one=\frac{(n+1)}{2^{2n}}\C[2n]{n}\Big(1-\frac1{(2n+1)^2}\Big)^n.$$
\item The probability
$P_{\one}=\sum_{k=n}^{2n}\C[2n]{k}\frac{(n+1)^kn^{2n-k}}{(2n+1)^{2n}}$ of takeover the firm in the post-toehold strategy is equal to
    $$P_1=\frac12+\frac1{2^{2n+1}}{\C[2n]{n}}\Big(1-\frac1{(2n+1)^2}\Big)^n+
\frac{n}{2^{2n}}\C[2n]{n}\int_0^{\frac1{2n+1}}(1-t^2)^{n-1}dt.$$
\item The price $X_1=\sum_{k=n}^{2n-1}\C[2n-1]{k}\frac{(n+1)^kn^{2n-1-k}}{(2n+1)^{2n-1}}$ for a share offered by the investor in the post-toehold tender offer can be calculated as
$$X_\one=\frac12+\frac{n}{2^{2n}}\C[2n]{n}\int_{0}^{\frac1{(2n+1)}}(1-t^2)^{n-1}dt.$$
 \end{enumerate}
\end{theorem}

\begin{proof} 1. To deduce the formula for the price $X_{\zero}$, we use the beta-equality (\ref{eq:beta}) with parameters $a=n+1$ and $b=n$. In this case we get the equality
\begin{equation}\label{eq1}
\sum_{k=n+1}^{2n}\C[2n]{k}x^k(1-x)^{2n-k}=(n+1)\C[2n]{n+1}\int_0^xt^n(1-t)^{n-1}dt=\\=n\C[2n]{n}\int_0^xt^n(1-t)^{n-1}dt.
\end{equation}
For $x=\frac12$, this equality turns into:
\begin{equation}
n\C[2n]{n}\int_0^\frac12t^n(1-t)^{n-1}=\sum_{k=n+1}^{2n}\C[2n]{k}\frac{1}{2^{2n}}
=\frac12\Big(-\C[2n]{n}\frac1{2^{2n}}+\sum_{k=0}^{2n}\C[2n]{k}\frac1{2^{2n}}\Big)
=\frac12\Big(1-\C[2n]{n}\frac1{2^{2n}}\Big)
\end{equation}
because $$\sum_{k=0}^{2n}\C[2n]{k}\frac1{2^{2n}}=\Big(\frac12+\frac12\Big)^n=1.$$

Then (\ref{eq1}) can be written as:
{
$$
\begin{aligned}
&\sum_{k=n+1}^{2n}\C[2n]{k}x^k(1-x)^{2n-k}=n\C[2n]{n}\Big(\int_0^{1/2}t^n(1-t)^{n-1}dt+
\int_{1/2}^{x}t^n(1-t)^{n-1}dt\Big)=\\
&=\frac12-\frac1{2^{2n+1}}{\C[2n]{n}}+n\C[2n]{n}\int_0^{x-\frac12}(\tfrac12+u)^n(\tfrac12-u)^{n-1}du=\\
&=\frac12-\frac1{2^{2n+1}}{\C[2n]{n}}+n\C[2n]{n}\int_0^{x-\frac12}(\tfrac12+u)(\tfrac14-u^2)^{n-1}du=\\
&=\frac12-\frac1{2^{2n+1}}{\C[2n]{n}}+n\C[2n]{n}\Big(\frac12\int_0^{x-\frac12}(\tfrac14-u^2)^{n-1}du-
\frac12\int_0^{x-\frac12}(\tfrac14-u^2)^{n-1}d(\tfrac14-u^2)\Big)=\\
&=\frac12-\frac1{2^{2n{+}1}}\C[2n]{n}+\frac{n}{2^{2n{-}1}}\C[2n]{n}\int_0^{x-\frac12}(1-(2u)^2)^{n{-}1}du-
\frac12{\C[2n]{n}}\big((x-x^2)^n-\tfrac1{4^n})=\\
&=\frac12+\frac{n}{2^{2n}}\C[2n]{n}\int_0^{2x-1}(1-t^2)^{n-1}dt-
\frac12{\C[2n]{n}}x^n(1-x)^n.
\end{aligned}
$$
}
For $x=\sigma_{\zero}=\frac12+\frac1{2(2n+1)}$ the latter formula yields the required formula for the price $X_{\zero}$:
$$
X_{\zero}=\sum_{k=n+1}^{2n}\C[2n]{k}\sigma_{\zero}^k(1-\sigma_{\zero})^{2n-k}=\frac12-
\frac1{2^{2n+1}}{\C[2n]{n}}\Big(1-\frac1{(2n+1)^2}\Big)^n+
\frac{n}{2^{2n}}\C[2n]{n}\int_0^{\frac1{2n+1}}(1-t^2)^{n-1}dt.
$$

2. By analogy we deduce the formula for the probability $$P_{\zero}=\sum_{k=n+1}^{2n+1}\C[2n+1]{k}\frac{(n+1)^kn^{2n+1-k}}{(2n+1)^{2n+1}}$$ of takeover the firm in the no-toehold strategy. Writing down the beta-equality (\ref{eq:beta}) for the parameters $a=n+1$ and $b=n+1$, we get
\begin{equation}\label{eq:P_0}
\sum_{k=n+1}^{2n+1}\C[2n+1]{k}x^k(1-x)^{2n+1-k}=(n+1)\C[2n+1]{n+1}\int_0^xt^n(1-t)^{n}dt.
\end{equation}
For $x=\frac12$, this equality turns into:
$$(n+1)\C[2n+1]{n+1}\int_0^\frac12t^n(1-t)^{n}=\sum_{k=n+1}^{2n+1}\C[2n+1]{k}\frac{1}{2^{2n+1}}=\frac12.$$
After suitable rearrangements, for $x=\sigma_{\zero}=\frac{n+1}{2n+1}$ the equality (\ref{eq:P_0}) transforms into the desired equality:
{$$\label{eq2:P0x}
\begin{aligned}
P_{\zero}=&\sum_{k=n+1}^{2n+1}\C[2n+1]{k}x^k(1-x)^{2n+1-k}=(n+1)\C[2n+1]{n+1}\Big(\int_0^{1/2}t^n(1-t)^{n}dt+
\int_{1/2}^{x}t^n(1-t)^{n}dt\Big)=\\
&=\frac12+(n+1)\C[2n+1]{n+1}\int_0^{x-1/2}(\tfrac12+u)^n(\tfrac12-u)^{n}du=
\frac12+\frac{(2n+1)}{2^{2n+1}}\C[2n]{n}\int_0^{2x-1}(1-t^2)^{n}dt=\\
&=\frac12+\frac{(2n+1)}{2^{2n+1}}\C[2n]{n}\int_0^{\frac1{2n+1}}(1-t^2)^{n}dt.
\end{aligned}
$$
}

3. The formula for the profits $\Pi_{\zero}=\Pi_\one=\C[2n]{n}\frac{n^n(n+1)^{n+1}}{(2n+1)^{2n}}$ follows from the observation that $$\frac{n^2+n}{(2n+1)^2}=\frac14\Big(1-\frac1{(2n+1)^2}\Big).$$

4. Taking into account that $\sigma_1=\sigma_0=\frac{n+1}{2n+1}$ and looking at the formula for $X_0$ proved in Theorem~\ref{formulas}(1), we see that
$$
\begin{aligned}
P_1&=\sum_{k=n}^{2n}\C[2n]{k}\sigma_1^k(1-\sigma_1)^{2n-k}=\C[2n]{n}\sigma_1^n(1-\sigma_1)^n+X_0=\\
&=\frac1{2^{2n}}{\C[2n]{n}}\Big(1-\frac1{(2n+1)^2}\Big)^n+\frac12-\frac1{2^{2n+1}}{\C[2n]{n}}\Big(1-\frac1{(2n+1)^2}\Big)^n+\frac{n}{2^{2n}}\C[2n]{n}\int_0^{\frac1{2n+1}}(1-t^2)^{n-1}dt=\\
&=\frac12+\frac1{2^{2n+1}}{\C[2n]{n}}\Big(1-\frac1{(2n+1)^2}\Big)^n+
\frac{n}{2^{2n}}\C[2n]{n}\int_0^{\frac1{2n+1}}(1-t^2)^{n-1}dt.
\end{aligned}
$$

5. The beta-equation (\ref{eq:beta}) written for $a=n$ and $b=n$ yields:
$$
\begin{aligned}
X_\one&=\sum_{k=n}^{2n-1}\C[2n-1]{k}\sigma_1^k(1-\sigma_1)^{2n-1-k}=n\C[2n-1]{n}\int_0^{\sigma_1}t^{n-1}(1-t)^{n-1}dt=\\
&=n\C[2n-1]{n}\int_0^{\frac12}t^{n-1}(1-t)^{n-1}dt+n\C[2n-1]{n}
\int_{\frac12}^{\sigma_1-\frac12}t^{n-1}(1-t)^{n-1}dt=\\
&=\sum_{k=n}^{2n-1}\C[2n-1]{k}\frac1{2^{2n-1}}+n\C[2n-1]{n}
\int_{0}^{\frac1{2(2n+1)}}(\tfrac12+u)^{n-1}(\tfrac12-u)^{n-1}du=\\
&=\frac12\sum_{k=0}^{2n-1}\C[2n-1]{k}\frac1{2^{2n-1}}+\frac{n}2\C[2n]{n}
\int_{0}^{\frac1{2(2n+1)}}(\tfrac14-u^2)^{n-1}du=\\
&=\frac12+\frac{n}{2^{2n-1}}\C[2n]{n}\int_{0}^{\frac1{2(2n+1)}}(1-4u^2)^{n-1}du=\frac12+\frac{n}{2^{2n}}\C[2n]{n}\int_{0}^{\frac1{(2n+1)}}(1-t^2)^{n-1}dt.\\
\end{aligned}
$$
\end{proof}

Using the formulas from Theorem~\ref{formulas}, one can derive the following lower and upper bounds for the parameters of our models, see \cite{BBKT} for details.

\begin{theorem}\label{bounds} The parameters of the models lie in the following intervals:
\begin{enumerate}
\itemindent23pt
\vskip5pt

\item[$X_{\zero}$:] $\;\displaystyle \frac12+\frac1{\sqrt{\pi n}}\Big(-\frac{1}{6n}-\frac1{64n^2}\Big)<X_{\zero}<\frac12+\frac1{\sqrt{\pi n}}\Big(-\frac{1}{6n}+\frac5{24n^2}\Big)$,
\vskip5pt

\item[$X_{\one}$:] $\;\displaystyle \frac12+\frac1{\sqrt{\pi n}}\Big(\frac12-\frac{5}{16n}+\frac1{48n^2}\Big)<X_{\one}<\frac12+\frac1{\sqrt{\pi n}}\Big(\frac12-\frac{5}{16n}+\frac1{12n^2}\Big)$,
\vskip5pt

\item[$P_{\zero}$:] $\;\displaystyle \frac12+\frac{1}{\sqrt{\pi n}}\Big(\frac12-\frac5{48n}+\frac1{16n^2}
\Big)<P_{\zero}<\frac12+\frac{1}{\sqrt{\pi n}}\Big(\frac12-\frac5{48n}+\frac6{16n^2}\Big)$,
\vskip5pt

\item[$P_{\one}$:]    $\;\displaystyle \frac12+\frac{1}{\sqrt{\pi n}}\Big(1-\frac{13}{24n}+\frac3{16n^2}\Big)<P_{\one}<\frac12+\frac{1}{\sqrt{\pi n}}\Big(1-\frac{13}{24n}+\frac4{16n^2}\Big),$

\vskip5pt
\item[$\Pi_{\zero},\Pi_\one$:] $\;\displaystyle
\frac{1}{\sqrt{\pi n}}\Big(n+\frac5{8}-\frac1{4n}\Big)<
\Pi_{\zero}=\Pi_\one<\frac{1}{\sqrt{\pi n}}\Big(n+\frac5{8}-\frac1{24n}+\frac1{3n^2}\Big)$.

\end{enumerate}
\end{theorem}

Looking at the bounds for the probabilities $P_\zero,P_\one$ and the prices $X_\zero,X_\one$ we can notice that $P_{\zero}<P_{\one}$ and $X_\zero<X_\one$. An estimation of the differences $P_{\one}-P_{\zero}$ and $X_\one-X_\zero$ is given in the following corollary of Theorem~\ref{bounds}.

\begin{corollary}\label{diff}
$$\frac{1}{\sqrt{\pi n}}\cdot\Big(\frac12-\frac{31}{48n}-\frac{3}{16n^2}\Big)<P_{\one}-P_{\zero}<\frac{1}{\sqrt{\pi n}}\cdot\Big(\frac12-\frac{31}{48n}+\frac{3}{16n^2}\Big)$$and
$$\frac{1}{\sqrt{\pi n}}\cdot\Big(\frac12-\frac{7}{48n}-\frac{3}{16n^2}\Big)<X_{\one}-X_{\zero}<\frac{1}{\sqrt{\pi n}}\cdot\Big(\frac12-\frac{7}{48n}+\frac{1}{24n^2}\Big).$$
\end{corollary}

\begin{remark} The difference $X_\one-X_\zero\approx\frac1{2\sqrt{\pi n}}$ can be interpreted as the price for the information that the investor possesses a toehold.
\end{remark}

The lower and upper bounds of Theorem~\ref{bounds} can be derived using the following lower and upper bounds for functions appearing in the formulas in Theorem~\ref{formulas}.

\begin{lemma}\label{elbound} For every $n\in\IN$ and a real number $x>0$ the following inequalities hold:
\begin{enumerate}
\item $1-x<1-x+\frac12x^2-\frac16x^3<e^{-x}<1-x+\frac12x^2$;
\item $1-x<\frac1{1+x}<1-x+x^2$;
\item $1-nx<(1-x)^n<1-nx+\frac{n(n-1)}2x^2$.
\end{enumerate}
\end{lemma}

\begin{lemma}\label{somebounds} The following lower and upper bounds hold for every $n\in\IN$:
\begin{enumerate}
\item $\quad\displaystyle \frac1{2n}-\frac1{4n^2}+\frac1{12n^3}<\frac1{2n+1}<\frac1{2n}-\frac1{4n^2}+\frac1{8n^3}$;
\item $\quad\displaystyle 1-\frac1{4n}+\frac1{8n^2}<{\Big(1-\frac1{(2n+1)^2}\Big)^n}<1-\frac1{4n}+\frac9{32n^2}$;
\item $\quad\displaystyle \frac1{2n}-\frac{7}{24n^2}+\frac{11}{48n^3}<\int_0^{\frac1{2n+1}}(1-t^2)^{n}dt<
    \frac1{2n}-\frac{7}{24n^2}+\frac{18}{48n^3}$;
\item $\quad\displaystyle \frac1{2n}-\frac{7}{24n^2}+\frac{5}{48n^3}< \int_0^{\frac1{2n+1}}(1-t^2)^{n-1}dt<\frac1{2n}-\frac{7}{24n^2}+\frac{12}{48n^3}$.
\end{enumerate}
\end{lemma}

\begin{lemma}\label{binom} The lower and upper bounds
$$\frac{4^n}{\sqrt{\pi n}}\Big(1-\frac1{8n}+\frac1{64n^2})<\binom{2n}{n}<\frac{4^n}{\sqrt{\pi n}}\Big(1-\frac1{8n}+\frac1{48n^2}\Big)$$hold for every $n\in\IN$.
\end{lemma}

These bounds on the central binomial coefficients can be derived from the following refined version of the Stirling formula for factorials, proved in  \cite{M65} and \cite{R55}.

\begin{lemma}\label{stirling} For every $n\ge 1$
$$n!=\sqrt{2\pi n}\Big(\frac ne\Big)^ne^{r_n}$$where
$$\frac1{12n}-\frac1{2^63n^3}<r_n<\frac1{12n}.$$
\end{lemma}

\section{Conclusion}

The analysis of our models witnesses that both strategies (toehold and no toehold) of taking over the firm with $2n+1$ shareholders yield the same profit $\Pi_0=\Pi_1\approx \sqrt{\frac{n}{\pi}}$ but the probability $P_1\approx \frac12+\frac1{\sqrt{\pi n}}$ of taking over for the toehold strategy is higher than the corresponding probability $P_0\approx\frac12+\frac1{2\sqrt{\pi n}}$ for the no-toehold strategy. Also the equilibrium price $X_1\approx\frac12+\frac1{2\sqrt{\pi n}}$ for a share offered by the investor in the tender offer announced after buying a toehold is higher that the corresponding price $X_0\approx \frac12-\frac1{6n\sqrt{\pi n}}$ in the tender offer without toehold.
The difference $X_\one-X_\zero\approx\frac1{2\sqrt{\pi n}}$ can be interpreted as the price for the information that the investor possesses a toehold.

\end{document}